\documentclass[10pt,conference]{IEEEtran}
\usepackage{amsmath}
\usepackage{amssymb}
\usepackage{upref}
\usepackage{amsfonts}
\usepackage{graphicx}
\usepackage{color}
\usepackage{epic}
\usepackage{eepic}
\usepackage{subfig}
\usepackage{verbatim}

\input{psfig}

\newcommand{\deff}{\mbox{$\stackrel{\rm def}{=}$}}

\newcommand{\field}[1]{\mathbb{#1}}

\newcommand{\F}{\field{F}}

\newcommand{\cC}{{\cal C}}
\newcommand{\cG}{{\cal G}}

\newcommand{\cM}{{\cal M}}

\newcommand{\cP}{{\cal P}}

\newcommand{\sP}{\cP}
\newcommand{\sG}{\cG}

\newcommand{\Gr}{\smash{{\sG\kern-1.5pt}_q\kern-0.5pt(n,k)}}
\newcommand{\Gfourk}{\smash{{\sG\kern-1.5pt}_q\kern-0.5pt(4k,2k)}}
\newcommand{\Gk}{\smash{{\sG\kern-1.5pt}_q\kern-0.5pt(n,k_1)}}
\newcommand{\Gkk}{\smash{{\sG\kern-1.5pt}_q\kern-0.5pt(n,k_2)}}
\newcommand{\Grtwo}{\smash{{\sG\kern-1.5pt}_2\kern-0.5pt(n,k)}}
\newcommand{\Gkone}{\smash{{\sG\kern-1.5pt}_q\kern-0.5pt(n,k_1)}}
\newcommand{\Gktwo}{\smash{{\sG\kern-1.5pt}_q\kern-0.5pt(n,k_2)}}
\newcommand{\Ps}{\smash{{\sP\kern-2.0pt}_q\kern-0.5pt(n)}}

\newtheorem{theorem}{Theorem}

\newtheorem{remark}{Remark}
\newtheorem{example}{Example}

\begin{document}

\title{Error Resilience in Distributed Storage\\ via Rank-Metric Codes}
%Adversarial Error Resilience in Distributed Storage \\ Using MRD Codes and MDS Array Codes}
\author{
\IEEEauthorblockN{Natalia Silberstein, Ankit Singh Rawat and Sriram Vishwanath}
\IEEEauthorblockA{LINC, Dept. of ECE, UT Austin \\
Austin, TX 78712, USA.\\
Email: \{natalys, ankitsr, sriram\}@austin.utexas.edu}
%\and
%\IEEEauthorblockN{Ankit Singh Rawat}
%\IEEEauthorblockA{Dept. of ECE,\\
%University of Texas at Austin,\\
%Austin, TX 78712, USA.\\
%Email: ankitsr@mail.utexas.edu}
%\and
%\IEEEauthorblockN{Sriram Vishwanath}
%\IEEEauthorblockA{Dept. of ECE,\\
%University of Texas at Austin,\\
%Austin, TX 78712, USA.\\
%Email: sriram@ece.utexas.edu}
}

\maketitle
%%%%%%%%%%%%%%%%%%%% Abstract %%%%%%%%%%%%%%%%%%%%%%%%%
\begin{abstract}
This paper presents a novel coding scheme for distributed storage systems containing nodes with adversarial errors.
The key challenge in such systems is the propagation of  erroneous data from a single corrupted node to the rest
of the system during a node repair process.
This paper presents a concatenated coding scheme which is based on two types of codes:
maximum rank distance (MRD) code as an outer code and  optimal repair maximal distance separable (MDS)
array code as an inner code. Given this, two different types of adversarial errors are considered: the first type considers an adversary that can replace the content of an affected node only once; while the second attack-type considers an adversary that can pollute data an unbounded number of times. This paper proves that the proposed coding scheme attains
a suitable upper bound on resilience capacity for the first type of error. Further, the paper presents mechanisms that combine this code with subspace signatures to achieve error resilience for the second type of errors. Finally, the paper concludes by presenting a construction based on MRD codes for optimal locally repairable scalar codes that can tolerate adversarial errors.

\end{abstract}
%*******************************************************************************
%*                                                                             *
%*                               Introduction                                  *
%*                                                                             *
%*******************************************************************************
\section{Introduction}
\label{sec:introduction}

%In light of exponential growth in the amount of data that is being generated, the issue of designing new storage mechanisms to handle this vast amount of data has grown as one of the primary challenges. The surge in the number of papers in this area over the past decade is a manifestation of the importance of this problem.

Distributed storage systems (DSS) are increasingly important resources today for users and businesses than ever before. Centralized storage is proving to be rapidly infeasible given the growing sizes of databases being stored; as well as the need to rapidly and reliably access them. Thus, generally, data is distributed and often replicated across multiple locations to enable ease-of-access and reslience to failures. However, replication can be inefficient in terms of storage space, and therefore, coding is useful in instilling resilience to node failures while reducing storage requirements over replication. Given the prevalence of single node failures in DSS  (a user exiting a P2P system, power outage in a single data center in the cloud), coding enables a single node to be repaired as soon as the node failure occurs in order to sustain the desired level of redundancy in the system.
%There are several reasons due to which instantaneous repair is desirable; one of which is to prevent permanent loss of data in the event of a catastrophic failure.

In order to repair the failed node, data is downloaded from surviving nodes and a function of this data is stored as the `restored' node. The amount of data downloaded in this repair process is called the {\em repair bandwidth} of the reconstruction process.  %A na\"{\i}ve strategy for node repair is to download all data from  surviving nodes to enable regeneration of the failed node. However, such an approach leads to a large repair bandwidth and consumes a vast amount of system resources (in terms of bandwidth and energy). Therefore, it is desirable to have a repair scheme that has as small a repair bandwidth as possible.
Since large repair bandwidth translate to consumption of a vast amount of system resources (in terms of bandwidth and energy), it is desirable to have a coding schemes that have as small a repair bandwidth as possible In \cite{dimakis}, Dimakis et al. establish an information theoretic lower bound on {repair bandwidth} for $(n,k)$ codes with maximum distance separable (MDS) property (i.e., any set of $k$ nodes can be used to reconstruct the data stored on a DSS). The work also presents a trade-off between  repair bandwidth and the amount of data stored on each node.

%\textit{Functional repair} is the one where the original failed node may not be replicated exactly, but to another that is {\em functionally} equivalent. \cite{wu} and \cite{wu2} present storage schemes (i.e., codes) that achieve the lower bound on {repair bandwidth}. An alternative and rather
A desirable notion of repair is \textit{exact repair}, where the regenerated data is an exact replica of what was stored on the failed node. The work in \cite{wu3, shah, suh1} present low rate codes, which achieve the lower bound derived in \cite{dimakis} when data is downloaded from all surviving nodes during node repair. %The coding schemes in \cite{wu3} and \cite{shah, suh1} work for $k < 3$ and $k \leq \frac{n}{2}$ respectively.
In \cite{rashmi}, Rashmi et al. design exact-repairable codes, which allow node repair to be performed by contacting $d \leq n-1$ surviving nodes. %These codes are optimal for all parameters $(n,k,d)$ at minimum bandwidth regeneration (MBR) point.
At the minimum storage regeneration (MSR) point of operation (see \cite{dimakis} for details), these codes also correspond to a low rate regime as their rate is upper bounded by $\frac{1}{2} + \frac{1}{2n}$. Recently, various researchers have presented high(er) rate codes for the MSR point which enable bandwidth-efficient exact repair. Along a similar view, \cite{dimitris11} presents codes for DSS with two parity nodes. In \cite{cadambe11} and \cite{bruck11}, permutation-matrix based codes are presented which are designed to achieve the bound on {repair bandwidth} for repair of systematic nodes for all $(n,k)$ pairs. \cite{bruck11_2} further generalizes the idea of \cite{bruck11} to get MDS array codes for DSS that allow optimal exact regeneration for parity nodes as well.
%All the schemes for high rate $k > n-k$ suffer from the fact that each node needs to store a large amount of data (exponential in $k$) at each node. %In their recent work, Cadambe et al. \cite{cadambe_asil} propose polynomial length codes in high rate regime.

While a bulk of existing literature in DSS  addresses the storage versus  repair bandwidth trade-off, another important issue that has recently received attention is the design of storage schemes that ensure security and reliability of the stored content against \textit{adversarial errors}~\cite{ho10,datta_byz,pawar, RSRK12}; and this issue of inducing reliability against adversarial errors is the main focus of our paper. The dynamic nature of DSS resulting from repeated node failures, triggering {\em node repairs}, makes the problem of dealing with erroneous nodes non-trivial. This is because a single corrupted node may subsequently pollute other nodes in   the DSS system during the  node repair process. In \cite{pawar}, Pawar et al. address the reliability issue in detail and derive upper bounds on the amount of data that can be stored on the system and reliably made available to a data collector when bandwidth optimal node repair is performed.% The authors consider two models for adversarial errors introduced in the storage nodes: 1) an omniscient adversary, who can observe all the nodes and knows the coding scheme employed by the system, 2) a limited knowledge adversary that can observe a maximum number of nodes throughout. In both error models, the adversary can control at most a fixed number of nodes and inject false information through these nodes during the entire operation of DSS.
 ~\cite{pawar} also presents coding strategies that achieve the upper bound for a particular range of system parameters, namely in the bandwidth-limited regime. A related but different problem of securing stored data against passive eavesdroppers is addressed in \cite{pawar, SRK11}.

%In this model, we adopt the notion of omniscient adversary and further classify these adversarial attacks in two classes:
In this paper, we study the notion of an omniscient adversary\footnote{In \cite{RSRK12}, Rashmi et al. show that product matrix codes~\cite{rashmi} can be used towards dealing with adversarial errors. This solution is again restricted to the regime $k \leq \frac{n+1}{2}$, at the MSR point of operation.}, which can observe all nodes and has full knowledge of the coding scheme employed by the system \cite{pawar}. As in~\cite{pawar}, we assume an upper bound on the number of nodes that can be controlled by such an omniscient adversary. %This upper bound is a system parameter that is used to design coding schemes to ensure reliable delivery of the original data to a data collector.
We classify adversarial attacks into two classes:
\begin{enumerate}
\item {\em Static errors:}  an omniscient adversary replaces the content of an affected node with nonsensical information \emph{only once}. The affected node %is unaware of this, and so it
     uses this \emph{same} polluted information during all subsequent repair and data collection processes. Static errors represent a common type of data corruption due to wear out of storage devices, such as latent disk errors or other physical defects of the storage media, where the data stored on a node is permanently distorted. %original data is consistently and permanently distorted.
\item {\em Dynamic errors:} an omniscient adversary may replace the content of an affected node, \emph{each time} the node is asked for the data during data collection or repair process. This kind of errors captures any malicious behaviour, hence is more difficult to manage in comparison to static errors.
\end{enumerate}
%In this paper,

 We present a novel concatenated coding scheme for DSS which provides resilience against these two classes of attacks. %The scheme is optimal under static error model as it attains the upper bound on the amount of data that can be stored reliably in DSS.
%?? In addition, this scheme is also optimal for the dynamic model, however, only  for specific choice of parameters. ??
 %(**{\em I think, may be reference to [13] confused the reviewer. He might have thought that we were aiming to achieve the bound in \cite{pawar} which at first look seems not true for static error model. Also, our scheme is optimal in static model only...not in dynamic error model. Somehow, this is not made clear here.}**).
 In our scheme, the content to be stored is first encoded using a maximum rank distance (MRD) code. The output of this outer code is further encoded using a maximum distance separable (MDS) array code, which can be repaired in a bandwidth efficient manner. Using an MRD code, which is an optimal rank-metric code, allows us to quantify the errors introduced in the system using their rank as opposed to their Hamming weights. As stated before, the dynamic nature of the DSS causes a large number of nodes to get polluted even by a single erroneous node, as  false information spreads from node repairs. Thus, a single polluted node infects many others, resulting in an error vector with a large Hamming weight. Using rank-metric codes can  help alleviate this problem as the error that a data collector has to handle has a known rank, and can therefore be corrected by an MRD code with a sufficient rank distance.
  %Hence, viewing errors in a DSS from a rank-metric perspective is advantageous from the point of view of resilience to static errors.
 Using an $(n,k)$ bandwidth efficient MDS array code as inner code facilitates bandwidth efficient  node repair in the event of a single node failure and allows the data collector to recover the original data from any subset of $k$ storage nodes. In this paper, we use exact-regenerating  bandwidth efficient codes operating at the minimum-storage regenerating (MSR) point~\cite{dimakis}. However, our construction works for any regenerating code.

The proposed coding scheme is directly applicable to the \textit{static error} model and is optimal in terms of amount of data that can be reliably stored on a DSS under static error model. The model with \textit{dynamic errors} is more complicated, as it allows a single malicious node to change its pollution pattern, and introduce an arbitrarily large error both in Hamming weight and in rank. For dynamic error model, we propose two solutions based on the concatenated coding scheme employed for static error model. One solution is to exploit the inherent redundancy in the encoded data due to outer code, i.e., an MRD code, and perform error free node repair even in the presence of adversarial nodes. This solution, namely na\"{\i}ve method, is optimal for a specific choice of parameters. Alternative solution combines our concatenated coding scheme with subspace signature based cryptographic schemes to control the amount (rank) of pollution (error) introduced by an adversarial node. We employ the signature scheme by Zhao et al. \cite{zhao}, which essentially reduces the dynamic error model to somewhat similar to static error model, and helps us bound the rank of error introduced by an adversarial node throughout its presence in the system. Note that hash function based solutions has previously been presented in the context of DSS to deal with errors \cite{ho10,pawar}. While promising, these hash functions based approaches provide only probabilistic guarantees for pollution containment.

Recently a new line of work has appeared in the context of DSS, which focuses on minimizing the number of surviving nodes that contribute data during node repair (locality). In \cite{GHSY11}, Gopalan et al. establish an upper bound on the minimum distance of scalar linear codes that have locality $r < k$. They further show the optimality of pyramid codes \cite{HCL07} with respect to the upper bound, with locality of the systematic nodes. \cite{PKLK12} gives a similar bound for codes that have a more general definition of locality. In \cite{PaDi12}, Papailiopoulos et al. generalize the bound in \cite{GHSY11} to vector codes (possibly nonlinear). In this paper, we also show that MRD codes can be utilized to construct optimal (in terms of minimum distance) locally repairable scalar codes, which can tolerate adversarial errors.

The rest of the paper is organized as follows: In Section~\ref{sec:preliminaries}, we first give a brief description of rank-metric codes along with Gabidulin MRD codes and the error model in rank-metric for these codes. Subsequently, we describe  MDS array codes and present two examples of  bandwidth efficient MDS array codes that are later used as inner codes in our construction. In Section~\ref{sec:construction}, we describe the construction of our storage scheme and prove its error resilience under the static error model. Further, we  present a few examples to illustrate our scheme and show that our codes attain the upper bound on resilience capability. In Section~\ref{sec:dynamic}, we address the dynamic error model. In Section~\ref{sec:locality}, we describe MRD codes based locally repairable storage schemes. Finally, we conclude with Section~\ref{sec:conclusion}.
% and conclude the paper.
 %Sec.~\ref{sec:conclude} finally concludes the paper and points out the direction of future work.

%\noindent
%{\em Notation:} A quick note on notation. We use boldface upper case letters for block matrices and boldface lower case letters for vectors. $A(i,j)$ represents $j^{th}$ element in $i^{th}$ row of a matrix $A$. $A(i,:)$ denotes $i^{th}$ row of $A$. $A^{T}$ is used to represent transpose. For a matrix $A$, rank$(A)$ denote the rank of $A$. We use the shorthand $[n]:=\{1,\dots,n\}$ to denote the set of positive integers up to $n$.

%*******************************************************************************
%*                                                                             *
%*           Section: Preliminaries                                            *
%*                                                                             *
%*******************************************************************************

\section{Preliminaries}
\label{sec:preliminaries}

%*******************************************************************************
\subsection{System Model}
\label{subsec:model}
Let $M$ be the size of a file to be stored in a distributed storage system with $n$ nodes.
All data symbols belong to a finite field $\F_{Q}$ of size $Q$. Each node contains
$\alpha$ symbols. A data collector reconstructs the original file by downloading the
data stored in any set of $k$ out of $n$ nodes. This property is called an \emph{MDS property}.
When a node fails, its content can be reconstructed by downloading $\beta$ symbols from
any $d\geq k$ surviving nodes and stored in a new node added to DSS instead of the failed node.
The amount of data needed for a node repair, denoted by $\gamma=\beta d$, is called the \emph{repair bandwidth}.
In~\cite{dimakis}, Dimakis et al. show a trade-off between  repair bandwidth $\gamma$ and the amount of data $\alpha$ stored on each node.
Two families of codes which attain two extremal points on the optimal tradeoff curve are called
\emph{minimum-storage regenerating (MSR) codes} and\emph{ minimum-bandwidth regenerating (MBR) codes}, respectively.
In the sequel we will focuss on the MSR codes. It was proved~\cite{dimakis} that for MSR codes $\alpha$ and $\gamma$ satisfy
\begin{equation}
\label{eq:trade-off}
(\alpha, \gamma)=\left(\frac{M}{k},\: \frac{Md}{k(d-k+1)}\right).
\end{equation}

%*******************************************************************************
\subsection{Rank-Metric Codes}
\label{subsec:rank-metric}

Rank-metric codes were introduced by Delsarte~\cite{Del78} and rediscovered
in~\cite{Gab85,Rot91}. These codes have
applications in different fields, such as space-time coding~\cite{LGB03}, random network coding~\cite{SKK08,EtSi09},
and  public key cryptosystems~\cite{GPT91}.
Our goal in this paper is to show that rank-metric codes are useful for error correction in distributed storage as well.

Let $\F_q$ be the finite field of size $q$. For two $N \times m$ matrices $A$ and $B$ over $\F_q$ the {\it
rank distance} is defined by
$$
d_R (A,B) \deff \text{rank}(A-B)~.
$$
An  $[N \times m,\varrho,\delta]$ {\it rank-metric code} $\cC$
is a linear code, whose codewords are $N \times m$ matrices
over $\F_q$; they form a linear subspace with dimension $\varrho$ of
$\F_q^{N \times m}$, and for each two distinct codewords $A$
and $B$, $d_R (A,B) \geq \delta$.
For an $[N \times m,\varrho,\delta]$ rank-metric code $\cC$ we have $\varrho \leq
\text{min}\{N(m-\delta+1),m(N-\delta+1)\}$
~\cite{Del78,Gab85,Rot91}. This bound, called Singleton bound for rank metric,
is achievable for all feasible parameters. The codes that achieve this bound are called {\it
maximum rank distance} (MRD) codes.
%Note, that for $m\leq N$, any MRD code is also an MDS code.

%****************************************************************************
\subsubsection{Gabidulin MRD Codes}
\label{subsubsec:Gabidulin}
An important family of MRD linear codes is presented by Gabidulin~\cite{Gab85}.
These codes can be seen as the analogs of Reed-Solomon codes for rank metric.
Let $m\leq N$.
A codeword in an $[N \times m, \varrho , \delta]$
rank-metric code $\cC$ can be represented by a
vector $\mathbf{c}=[c_1 , c_2 , \ldots , c_{m}]$, where $c_i \in \F_{q^N}$, since $\F_{q^N}$
can be viewed as an $N$-dimensional vector space over $\F_q$.
Let $g_i\in \F_{q^N}$, $1\leq i\leq m$,  be linearly independent over $\F_q$.
The generator matrix $\cG$ of an  $[N \times m,\varrho,\delta]$
Gabidulin MRD code is given by

\begin{small}
$$
\cG=\left[\begin{array}{cccc}
g_{1} & g_{2} & \ldots & g_{m}\\
g_{1}^{[1]} & g_{2}^{[1]} & \ldots & g_{m}^{[1]}\\
%g_{1}^{[2]} & g_{2}^{[2]} & \ldots & g_{m}^{[2]}\\
\vdots & \vdots & \ddots & \vdots\\
g_{1}^{[K-1]} & g_{2}^{[K-1]} & \ldots & g_{m}^{[K-1]}\end{array}\right],
$$
\end{small}
where $K=m-\delta+1$, $\varrho=NK$, and  $[i]=q^{i}$.

In the similar way as Reed-Solomon codes, Gabidulin codes also have an equivalent interpretation as
evaluation of polynomials, however, for Gabidulin codes the special family of polynomials,
called \emph{linearized polynomials}, is evaluated on a set of linearly independent points over the base field $\F_q$.
A linearized polynomial $f(x)$ over $\F_{q^N}$ of $q$-degree $n$ has the form $f(x)=\sum_{i=0}^{n}a_ix^{q^i}$,
where $a_i\in \F_{q^N}$, and $a_{n}\neq 0$. It is easy
to see that evaluation of a linearized polynomial is an $\F_{q}$-linear
transformation from $\F_{q^N}$ to itself, i.e., for any $\alpha, \beta \in \F_q$ and
$a,b\in\F_{q^N}$, we have $f(\alpha a+\beta b)=\alpha f(a)+\beta f(b)$~\cite{MWSl78}. Therefore, a codeword in Gabidulin code
$\mathcal{C}$ can be defined as $\mathbf{c} = [f(g_1),f(g_2),\ldots, f(g_m)]$, where $f(x)$ is the
linearized polynomial  of $q$-degree $K-1$ with coefficients given by the information message, and $g_1,\ldots,g_m\in \F_{q^N}$
are linearly independent over $\F_q$~\cite{Gab85}.

%We will use this family of MRD codes in our construction of codes for distributed storage.

%****************************************************************************
\subsubsection{Rank Errors and Rank Erasures Correction}
\label{subsubsec:rank error}

Let $\cC\subseteq \F_{q^N}^m$ be a Gabidulin MRD code with minimum distance $\delta$.
Let $\mathbf{c}\in \cC$ be the transmitted
codeword and let $\mathbf{r}=\mathbf{c}+\mathbf{e_{total}}$ be the received word.
The code $\cC$ can correct any vector error of the form $\mathbf{e_{total}}=\mathbf{e_{error}}+\mathbf{e_{erasure}}=
(e_1\mathbf{u_{1}}+ \ldots +e_t\mathbf{u_{t}})+ (r_1\mathbf{v_{1}}+ \ldots +r_s\mathbf{v_{s}})$ as
long as $2t+s\leq \delta-1$. The  first part $\mathbf{e_{error}}$ is called  a \emph{rank error}
of rank $t$, where $e_i\in\F_{q^N}$ are linearly independent over the base field $\F_q$,
unknown to the decoder, and $\mathbf{u}_i\in\F_q^m$
are linearly independent vectors of length~$m$, unknown to the decoder.
The second part $\mathbf{e_{erasure}}$ is called a\emph{ rank erasure}, where
$r_i\in\F_{q^N}$ are linearly independent over the base field $\F_q$,
unknown to the decoder, and $\mathbf{v}_i\in\F_q^m$
are linearly independent vectors of length~$m$,  and known to the decoder.
Note, that in the matrix form we can write
\begin{equation}
\label{eq:rank error}
 \mathbf{e_{total}}=\left[e_1\ldots e_t\right]\left[\begin{array}{c}
                          \mathbf{u_1}\\
                           \hline
                            \vdots \\
                            \hline
                           \mathbf{u_t}
                          \end{array}
\right]
+\left[r_1\ldots r_s\right]\left[\begin{array}{c}
                            \mathbf{v}_1 \\
                            \hline
                            \vdots \\
                            \hline
                            \mathbf{v}_s
                          \end{array}
\right].
\end{equation}
%\begin{equation}
%\label{eq:rank error}
% \mathbf{e_{total}}=\left[\mathbf{u}_1\ldots \mathbf{u}_t\right]\left[\begin{array}{c}
%                            e_1 \\
%                            \vdots \\
%                            e_t
%                          \end{array}
%\right]+
%\left[\mathbf{v}_1\ldots \mathbf{v}_s\right]\left[\begin{array}{c}
%                            r_1 \\
%                            \vdots \\
%                            r_s
%                          \end{array}
%\right].
%\end{equation}
%Decoding algorithms for rank-metric codes are provided in~\cite{Gab85,SiKs09}.
\cite{Gab85,SiKs09} present decoding algorithms for rank-metric codes.
%****************************************************************************
\begin{figure*}[!t]
% ensure that we have normalsize text
\normalsize

%\begin{figure}
  \centering
 % \vspace{0.1in}
  \subfloat[]{\label{fig:zigzag_des}\scalebox{0.259}{\input{zigzag_new.pstex_t}}}
  \hspace{0.3in}
  \subfloat[]{\label{fig:zigzag_repair}\scalebox{0.25}{\input{zigzag_repair.pstex_t}}}

  \caption{Illustration of the second node repair process in $(5,3)$ Zigzag code: (a) for error free system, (b) for system with erroneous information at the first storage node.}
  \label{fig:zigzag}
%\end{figure}

% IEEE uses as a separator
\hrulefill
% The spacer can be tweaked to stop underfull vboxes.
%\vspace*{2pt}
\end{figure*}

%*******************************************************************************
\subsection{MDS Array Codes for Distributed Storage}
\label{subsec:MDS Array Codes}

A linear \textit{array code} $C$ of dimensions $\alpha\times n$ over $\F_q$ is defined as a linear subspace of $\F_q^{\alpha n}$. Its
minimum distance $d_{\min}$ is defined as the minimum Hamming distance over  $\F_{q^\alpha}$, when
we consider the codewords of $C$ as vectors of length $n$ over $\F_{q^\alpha}$. An array code $C$ is
called an $(n,k)$ \textit{maximum distance separable} (MDS) code if $|C|=q^{\alpha k}$, where $k=n-d_{\min}+1$~\cite{BlRo99, CaBr06}.
Note, that an MRD code is also an MDS array code.

Let $\mathbf{x}=[\mathbf{x}_1,\mathbf{x}_2,\ldots,\mathbf{x}_k]\in \F_q^{\alpha k}$ be an information vector, $\mathbf{x}_i \in \F_q^\alpha$
is a block of size $\alpha$, for all $1\leq i\leq k$.
These $k$ blocks are encoded into $n$ encoded blocks $\mathbf{y}_i\in \F_q^\alpha$, $1\leq i\leq n$, stored in $n$ nodes
 of size $\alpha$,  in the following way:
$$\mathbf{y}=\mathbf{x}\mathbf{G},
$$
where $\mathbf{y}=[\mathbf{y}_1,\mathbf{y}_2,\ldots,\mathbf{y}_n]$ and the generator matrix $\mathbf{G}$ is an $k\times n$ block matrix with blocks of size $\alpha\times\alpha$ given by:

\begin{small}
$$\mathbf{G}=\left[\begin{array}{cccc}
            A_{1,1} & A_{1,2} & \ldots & A_{1,n}  \\
            A_{2,1} & A_{2,2} & \ldots & A_{2,n}  \\
            \vdots & \vdots & \ddots & \vdots \\
            A_{k,1} & A_{k,2} & \ldots & A_{k,n}
          \end{array}
\right].
$$
\end{small}
An array code $C$ \textit{has an MDS property} if any blocks sub matrix of $\mathbf{G}$ of size $k\times k$ is of the full rank.
In other words,  if an $(n,k)$ MDS code is used to store data in a system, and
any set of  $n-k$ storage nodes fails, the original data can be recovered from the $k$ surviving nodes.

We say that an MDS code satisfies \textit{optimal repair property}, if a single failed node can be repaired by downloading
$\alpha d/(d-k+1)$ elements from every node of any $d$-subset of surviving nodes~\cite{dimakis}.

\subsubsection{Examples of Optimal Repair MDS Array Codes}
\label{subsubsec:MDS examples}

In the following, we present two examples of the optimal repair MDS array codes for DSS, which we will use further for illustration of our coding scheme.
Due to space constraints, we only describe the MDS codes used in the examples. For general constructions, interested readers may refer to the respective papers that present these codes.

\begin{example}
\label{ex:zig-zag}
{(5,3) Zigzag code}~\cite{bruck11}.
This class of MDS array codes \cite{bruck11} is  based on generalized permutation matrices.
%Though Zig-zag codes are not repair optimal of parity nodes like closely related codes proposed in \cite{bruck11_2}, we prefer to choose example from Zig-zag codes as these codes suffice to make our point and all the conclusions drawn here directly extend the codes in \cite{bruck11_2} as well.
For the $(5,3)$ Zigzag code presented in Fig.~\ref{fig:zigzag}, the first three nodes are systematic nodes
which store the data $[c_1, c_2,\ldots, c_{12}]$, with $\alpha=4$.  The block generator matrix for this code is given by

\begin{small}
$$
\textbf{G}=\left[\begin{array}{ccccc}
               I & \emph{0} & \emph{0} & I & I \\
               \emph{0} & I & \emph{0} & I & A_2 \\
               \emph{0} & \emph{0} & I & I & A_3
             \end{array}
\right],
$$
$$
A_2 = \left[\begin{array}{cccc}
              0 & 0 & 1 & 0 \\
              0 & 0 & 0 & 1 \\
              2 & 0 & 0 & 0 \\
              0 & 2 & 0 & 0
            \end{array}
\right],
A_3=\left[\begin{array}{cccc}
              0 & 1 & 0 & 0 \\
              2 & 0 & 0 & 0 \\
              0 & 0 & 0 & 2 \\
              0 & 0 & 1 & 0
            \end{array}
\right],
$$
\end{small}
where $I$ and $\emph{0}$ denote the identity matrix and all-zero matrix, respectively.
Fig.~\ref{fig:zigzag_des} describes node repair process  for $(5,3)$ Zigzag code. When the second node fails, the newcomer node downloads the symbols from the shaded locations at the surviving nodes.
\end{example}

\begin{example}
\label{ex:hadamard}
{$(5,3)$ Hadamard Design codes} \cite{dimitris11}.
This class of MDS array codes employs interference alignment  strategies in order to perform node repair.
%These Hadamard design based codes have been shown to be bandwidth repair efficient only when $2$ parity nodes are involved.
In the $(5,3)$ example presented in Fig.~\ref{fig:hadamard}, the first three nodes are systematic nodes which
store the data, $\mathbf{y}_i = [c_{(i-1)\alpha +1},\ldots, c_{i\alpha}]$, $1\leq i\leq 3$.
The block generator matrix for a $(5,3)$ Hadamard design based code is given by

\begin{small}
$$
\textbf{G}=\left[\begin{array}{ccccc}
               I & \emph{0} & \emph{0} & I & A_{1,5} \\
               \emph{0} & I & \emph{0} & I & A_{2,5} \\
               \emph{0} & \emph{0} & I & I & A_{3,5}
             \end{array}
\right],
$$
\end{small}
where $A^T_{i,5}=a_iX_i+b_iX_4+I$, $1\leq i\leq 3$,
 $X_i = I_{2^{i-1}}\otimes\text{blkdiag}\left(I_{\frac{\alpha}{2^i}}, -I_{\frac{\alpha}{2^i}}\right)$, $\alpha = 2^{4}$, and $a_i$ and $b_i \in \mathbb{F}_q$
satisfy $a_i^2-b_i^2 = -1$. The process of the second node repair is illustrated in  Fig.~\ref{fig:hadamard_des}. During this process, the newcomer downloads $\mathbf{y}_4V$ and $\mathbf{y}_5V$ from node $4$ and $5$, respectively (see $D_1$ in Fig.~\ref{fig:hadamard_des}). The newcomer uses $\widehat{V}$ and $\widetilde{V}$ as repair matrices corresponding to node $1$ and $3$, respectively, where $\widehat{V}$ and $\widetilde{V}$ are some basis for the column-space of $[V~ A_{1,5}V]$ and $[V~ A_{3,5}V]$, respectively. The information downloaded from node $1$ and $3$ is used to cancel the interference terms (contribution of $\mathbf{y}_1$ and $\mathbf{y}_3$ in $D_1$). After interference mitigation, a linear system of equations is solved to get $\mathbf{y}_2$.
\end{example}

%*******************************************************************************
%*                                                                             *
%*                  The Construction                                           *
%*                                                                             *
%*******************************************************************************
\section{Code Construction for Static Error Model}
\label{sec:construction}

In this section we present our coding scheme and prove its error tolerance under the static error model. We illustrate the idea by using examples from Sec.~\ref{subsubsec:MDS examples} and prove that this construction is optimal for the static error model.

\subsection{The Construction}

Let $\mathbb{M}\in \F_q^{KN}$ denote a file of size $M=KN$, $K<N$. $\mathbb{M}$ is partitioned into $K$ parts of size $N$ each.
We form an $N\times K$ matrix $\cM$ over $\F_q$, where an $i$th part of $\mathbb{M}$ forms an $i$th column of $\cM$, for all $1\leq i\leq K$.
Let $\cC$ be an $[N \times m,\varrho=NK,\delta=m-K+1]$ Gabidulin MRD code, with $m\leq N$. Let $\mathbf{c}_{\cM}\in \F_{q^N}^m$ be the codeword
in $\cC$ which corresponds to the information matrix $\cM$.
Let $\alpha, k$ be positive integers such that $m=\alpha k$ and let $Q=q^N$.
Let $C$ be an $(n,k)$ optimal repair MDS array code of dimensions $\alpha \times n$.
We partition the vector $\mathbf{c}_{\cM}$ into $k$ parts of size $\alpha$ and form $n$ nodes of size $\alpha$ each,
according to the encoding algorithm of the code $C$.
Note, that we use an MDS array code
over $\F_q\subseteq \F_{q^N}=\F_Q$, i.e., its generator matrix is over $\F_q$, and during the process of node
repair, a set of surviving nodes transmits linear combinations of the stored elements  with the coefficients from $\F_q$.

The following theorem shows that if $2t\alpha+1\leq \delta$ the system
tolerates up to $t$ erroneous nodes.

\begin{theorem}Let $t$ be the number of erroneous nodes in the system based on concatenated MRD and optimal repair MDS array codes.
If $2t\alpha+1\leq \delta$, then the original data can be recovered from any $k$ nodes.
\end{theorem}

\begin{proof}
Let $\mathbf{c}_{\cM}\in \F_{Q}^m$ be the codeword
in $\cC$ which corresponds to an information matrix $\cM$, and let
$[\mathbf{x}_1,\mathbf{x}_2,\ldots,\mathbf{x}_k]$, $\textbf{x}_i\in\F_{Q}^{\alpha}$, be
the partition of $c_{\cM}$ into $k$ parts of size $\alpha$ each.
Let $[\mathbf{y}_1,\mathbf{y}_2,\ldots,\mathbf{y}_n]$, $\textbf{y}_i\in\F_{Q}^{\alpha}$ be the encoded blocks stored in $n$ nodes.

Let $S=\{i_1,i_2,\ldots,i_t\}$  be the set of indices  of the erroneous nodes.
Hence  the  $i_j$th node, $i_j\in S$, contains  $\sum_{\ell=1}^{k}\mathbf{x}_{\ell}A_{\ell,i_{j}}+\mathbf{e}^{i_j}$,
where $\mathbf{e}^{i_j}=[e^{i_j}_1,e^{i_j}_2,\ldots,e^{i_j}_{\alpha}]\in \F_{Q}^\alpha$ denotes an adversarial error introduced by the $i_j$th node.
%During the process of repair of failed nodes, these errors propagate
When the failed nodes are being repaired, the errors from adversarial nodes propagate
to the repaired nodes. In particular, $\ell$th node, $1\leq \ell \leq n$, contains
$\sum_{j=1}^{k}\mathbf{x}_jA_{j,\ell}+\sum_{j=1}^{t}\mathbf{e}^{i_j}B_{\ell}^{i_j}$,
where $B_{\ell}^{i_j}\in \F_q^{\alpha\times \alpha}$ represents the propagation of error $\mathbf{e}^{i_j}$ and
depends on the specific choice of an MDS array code.
Suppose a data collector contacts a subset $D$ with $k$ nodes and
downloads $\sum_{j=1}^{k}\mathbf{x}_jA_{j,i}+\sum_{j=1}^t \mathbf{e}^{i_j}B_{i}^{i_j}$ from any $i$th node, $i\in D$.
If these $k$ nodes are all systematic nodes, then
we obtain $[\mathbf{x}_1,\mathbf{x}_2,\ldots,\mathbf{x}_k]+\mathbf{eB}$, where $\mathbf{B}$ is the blocks matrix of
size $t\alpha\times k\alpha$ over $\F_q$  given by

\begin{small}
$$\mathbf{B}=\left[
      \begin{array}{cccc}
        B_1^{i_1} &  B_2^{i_1} & \ldots &  B_k^{i_1} \\
        B_1^{i_2} &  B_2^{i_2} & \ldots &  B_k^{i_2} \\
        \vdots & \vdots & \ddots & \vdots \\
        B_1^{i_t} &  B_2^{i_t} & \ldots &  B_k^{i_t} \\
      \end{array}
    \right]
$$
\end{small}
and $\mathbf{e}=[\mathbf{e}^{i_1}, \mathbf{e}^{i_2},\ldots, \mathbf{e}^{i_t}]$.
Otherwise, we obtain $[\mathbf{x}_1,\mathbf{x}_2,\ldots,\mathbf{x}_k]+\mathbf{eB'}$, where the block matrix
$\mathbf{B}'\in\F_q^{t\alpha\times k\alpha}$ represents the coefficients of $\mathbf{e}$ obtained by the decoding of the code $C$.
Since  the rank of $\mathbf{e}$ over $\F_q$ is at most $t\alpha$, and $\delta\geq 2t\alpha+1$, the MRD code $\cC$ can correct this error.
\end{proof}

\begin{figure*}[!t]
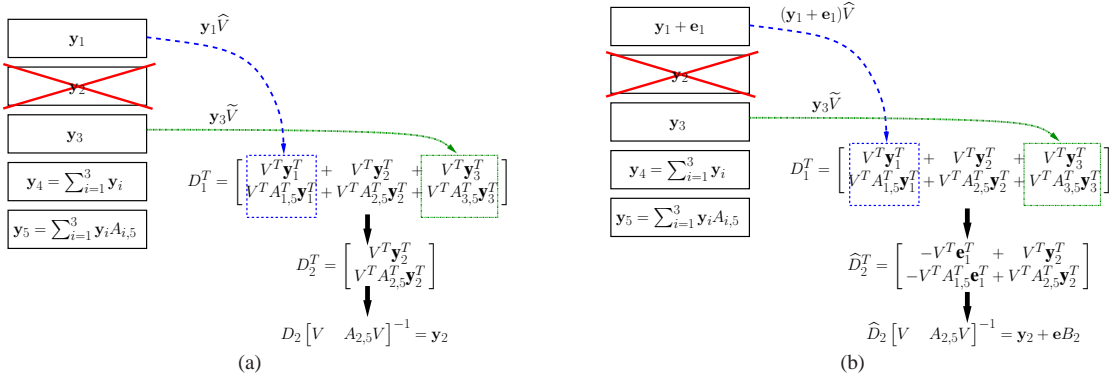

% ensure that we have normalsize text
\normalsize

%\begin{figure}
  \centering
%  \vspace{0.1in}
  \subfloat[]{\label{fig:hadamard_des}\scalebox{0.25}{\input{hadamard_new.pstex_t}}}
  \hspace{0.5in}
  \subfloat[]{\label{fig:hadamard_repair}\scalebox{0.25}{\input{hadamard_repair.pstex_t}}}

  \caption{Illustration of node repair in $(5,3)$ Hadamard design based codes: (a) in error free system, (b) in the presence of error at the first storage node.}
  \label{fig:hadamard}
%\end{figure}'

% IEEE uses as a separator
\hrulefill
% The spacer can be tweaked to stop underfull vboxes.
%\vspace*{2pt}
\end{figure*}

\vspace{-0.1in}

\vspace{0.5cm}

\subsection{Examples}
\label{subsec:dss_dynamicsExamples}
In this subsection we illustrate the idea of the construction for the case where an adversary pollutes the information stored on a single storage node. We demonstrate  that the rank of the error introduced by an adversary does not increase due to node repair dynamics under  the static error model. Hence, a  data collector can recover the correct original information using a decoder for an MRD code. It is important to note that in our construction, any optimal repair MDS array code from DSS literature can be used as the inner code. In this subsection, we illustrate the idea of our construction with the help of two examples drawn from two different classes of optimal repair MDS array codes for DSS, presented in Section~\ref{subsubsec:MDS examples}.

\begin{example} Let $C$ be the (5,3) Zigzag code from Example~\ref{ex:zig-zag}. Its
first three systematic nodes store a codeword $\textbf{c}=[c_1, c_2,\ldots, c_{12}] \in \mathbb{F}_{Q}^{12}$, $Q=q^N$, from Gabidulin MRD code, which is obtained by encoding the original data.  The content stored in $i$th systematic node, $1\leq i\leq 3$, is $\mathbf{y}_i = [c_{(i-1)\alpha +1},\ldots, c_{i\alpha}] \in \mathbb{F}_{Q}^{\alpha}$. Note that $\alpha = 4$ in this example.
Let us assume that an adversary attacks the first storage node and introduces erroneous information. The erroneous information at the first node can be modeled as $\mathbf{y}_1 + \mathbf{e} = [c_1, c_2, c_3, c_4]+[e_1, e_2, e_3, e_4]$. Now assume that the second node fails. The system is oblivious to the presence of pollution at the first node, and employs an exact regeneration strategy to reconstruct the second node. The reconstructed node downloads the symbols from the shaded locations at the surviving nodes, as described in Fig.~\ref{fig:zigzag_repair}, and solves a linear system of equations to obtain $[c_5, c_6, c_7, c_8] + [-e_1, -e_2, -2^{-1}e_1, -2^{-1}e_2]$, where $2^{-1}$ denotes the inverse element of $2$ in $\mathbb{F}_q$, $q\geq 3$. Now assume that a data collector accesses the first three nodes in an attempt to recover the original data. The data collector now has access to $\widetilde{\mathbf{c}} = \mathbf{c} + \mathbf{e}[I,~B_2,~0]$, where

\begin{small}
\begin{equation}
B_2 = \left[ \begin{array}{cccc}-1&0&-2^{-1}&0\\0&-1&0&-2^{-1}\\0&0&0&0\\0&0&0&0\\ \end{array}\right].
\end{equation}
\end{small}
Note that $\widetilde{\mathbf{c}}$ contains an error of rank at most four. Therefore, the original MRD codeword $\mathbf{c}$ and subsequently the original information can be recovered, using an MRD code with rank distance at least nine.
\end{example}

\begin{example}
Let $C$ be a (5,3) Hadamard design based code, described in Example~\ref{ex:hadamard}. Its
first three nodes store $\mathbf{y}_i = [c_{(i-1)\alpha +1},\ldots, c_{i\alpha}] \in \mathbb{F}_{Q}^{\alpha}$, $1\leq i\leq 3$, where $\textbf{c}=[c_1, c_2,\ldots, c_{3\alpha}] \in \mathbb{F}_{Q}^{3\alpha}$ is a codeword of a Gabidulin MRD code, which is obtained by encoding the original data.
Suppose an adversary modifies the information stored at the first node to $\mathbf{y}_1 + \mathbf{e} = [c_1,\ldots, c_{\alpha}] + [e_1,\ldots, e_{\alpha}]$. When the second node fails, a newcomer, unaware of the presence of error at the first node, employs the interference alignment based strategy described in Example~\ref{ex:hadamard} and depicted in  Fig.~\ref{fig:hadamard}.
%In order to exactly regenerate the failed information, the newcomer downloads $V^T\mathbf{y}_4$ and $V^T\mathbf{y}_5$ from node $4$ and $5$ respectively (see $D_1$ in Fig.~\ref{fig:hadamard}). The newcomer uses $\widehat{V}^T$ and $\widetilde{V}$ as repair matrices corresponding to node $1$ and $3$ respectively. Where $\widehat{V}^T$ and $\widetilde{V}^T$ are some basis for the row-space of $[V~ A_{5,1}^TV]^T$ and $[V~ A_{5,3}^TV]^T$ respectively. The information downloaded from node $1$ and $3$ is used to cancel the interference terms (contribution of $\mathbf{y}_1$ and $\mathbf{y}_3$ in $D_1$).
 After interference mitigation, a linear system of equations is solved to obtain $\mathbf{y}_2 + \mathbf{e}B_2$. Assuming that a data collector contacts the first three nodes, it receives $[\mathbf{y}_1,~ \mathbf{y}_2,~ \mathbf{y}_3] + \mathbf{e}[I,~ B_2,~ 0]$ which contains an error of rank at most $\alpha$. This allows the recovery of uncorrupted information using an MRD code of sufficient minimum rank distance.
\end{example}

%*******************************************************************************
\subsection{Code Parameters}
\label{subsec: parameters}

The upper bound on the amount of data that can be stored  reliably in the general system with
$t < \frac{k}{2}$ corrupted nodes, called \textit{resilience capacity} and denoted by $C(\alpha, \beta)$,  was presented by Pawar et al.~\cite{pawar}. This bound is given by
\begin{equation}\label{bound}
C(\alpha, \beta)\leq \sum_{i=2t+1}^{k}\min\{(d-i+1)\beta,\alpha\}.
\end{equation}
The authors  in~~\cite{pawar} provided the explicit
construction of the codes that attain this bound, for bandwidth-limited regime.
However, this construction has practical limitations for large values of $t$ since the decoding algorithm
presented in~\cite{pawar} is exponential in $t$.

 Next, we show that our constructed codes attain  bound~(\ref{bound}) and thus, are optimal. In addition, the decoding of codewords in the construction presented in our paper  is efficient since it is based on two efficient decoding algorithms: one, for an MDS array code, and two, for a Gabidulin code.

Let the parameters $K, N, m, \delta, \alpha, k, n$ be as described in the construction of Section~\ref{sec:construction}.
Then $m=K+\delta-1$ and  $m=\alpha k$.
Let $t$ be an integer such that $\delta = 2t\alpha+1$. Then $\alpha k=K+\delta-1=K+2\alpha t$, and hence $K=\alpha (k-2t)$.

Now we compare this result with the bound~(\ref{bound}).
Let $C$ be an MDS array code with optimal repair property. Then  $\beta =\frac{\alpha}{d-k+1}$.
Therefore, we can rewrite bound~(\ref{bound}) as follows:
%\begin{equation*}
\begin{flalign*}
\label{bound 2}
&C(\alpha, \beta=\frac{\alpha}{d-k+1})\leq&\\
&\sum_{i=2t+1}^{k}\min\{(d-i+1)\frac{\alpha}{d-k+1},\alpha\}=
\alpha(k-2t).&
\end{flalign*}
%\end{equation*}
Thus, the coding scheme proposed in this paper achieves the bound in (\ref{bound}).

\begin{remark}
%Despite the fact that static error model is a restricted error model in comparison with the error model addressed in \cite{pawar}, the upper bound~(\ref{bound}) still holds for static model. In \cite{pawar}, Pawar et al. obtain the upper bound by evaluating a cut of information flow graph corresponding to a particular node failure sequence, pattern of nodes under adversarial attack and data collector. This information flow graph also arises in the static error model, and its cut gives an upper bound on the flow, i.e., the amount of data that can be reliably stored on DSS, under static model as well. As shown in this paper, this bound is tight for static model at MSR point, which might not be the case for general error model of \cite{pawar}.
Although the static error model is less general than
the  model considered in~\cite{pawar}, the upper bound~(\ref{bound})
applies to the static model as well. Pawar et al.~\cite{pawar}
obtain this upper bound by evaluating a cut of the information flow graph
corresponding to a particular node failure sequence,
pattern of nodes under adversarial attack and data collector.
This information flow graph is also valid in the context of the static error model.
Consequently its cut, which provides an upper bound on the information
flow and represents the amount of data that can be reliably stored on
DSS, is correct under the static model.
As shown in this paper, this bound is tight for static model at MSR point,
which might not be the case for the general error model considered
in~\cite{pawar}.
\end{remark}

\begin{remark}
Recently, Rashmi et al.~\cite{RSRK12} considered a scenario,
referred as `\emph{erasure}',  where some nodes which are supposed to provide data during node repair become unavailable.
It is easy to see from~(\ref{eq:rank error})
that our construction can also correct such erasures, as long as the minimum distance of the corresponding MRD code is large enough.
The codes obtained by our construction also attain the bound on the capacity derived in~\cite{RSRK12}.
However, while our construction works with any MSR code and in
particular with an MSR code with high rate, it provides a solution for  a restricted error model.
%In a recent work~\cite{RSRK12}, Rashmi et al. consider the scenario where some nodes which are supposed to provide data during node repair become unavailable. This issue is referred in~\cite{RSRK12} as `\emph{erasure}'. It is straightforward to see from~(\ref{eq:rank error}) that our construction for error correction can also correct such erasures, as long as the minimum distance of the corresponding MRD code is large enough. The codes obtained by our construction also attain the bound on the capacity derived in~\cite{RSRK12}. However, while our construction is applicable to any MSR code and in particular to an MSR code with high rate, our construction provides the solution for a restricted model of errors.
\end{remark}

%*******************************************************************************
%*                                                                             *
%*                  Dynamic Errors                                             *
%*                                                                             *
%*******************************************************************************
\section{Dynamic Error Model}
\label{sec:dynamic}
In this section, we consider the problem of designing coding schemes for DSS that work under dynamic error model. In this attack model, an adversary can have access to at most $t$ storage nodes throughout the life span of DSS and can potentially send different data each time a node under its control is meant to send data during node repair as well as data reconstruction process.

Note that coding scheme proposed for the static error model hinges on the fact that each time an attacked node is requested for the data to be sent, it sends some linear combinations of the data that has been modified on it by adversary, which the adversary is allowed to do only once. Therefore, the rank of the error that a single node under static attack causes, throughout the operation of DSS, is bounded above by $\alpha$. This is not the case under the dynamic error model as a single attacked node can inject an error of large rank if it is utilized in multiple node repairs, which may render the data stored on DSS useless.

Towards this model, some results are presented in \cite{pawar} and \cite{RSRK12}. As we have discussed earlier, the coding scheme proposed in \cite{pawar} does not have an efficient decoding during the data reconstruction process and it works specifically with bandwidth efficiently repairable codes at MBR point, however we focus on MSR point in this paper. The coding scheme of \cite{RSRK12} deals with the dynamic error model at MSR point, but there scheme works only for low rate, i.e., $2k \leq n+1$. In coding scheme proposed in \cite{RSRK12}, an adversarial node does not inject errors into any other node during node repair, as this scheme allows newcomer to perform error-free exact repair.

Next, we present two solutions to deal with attack under the dynamic error model. The first solution aims to correct errors during the node repair process by using a simple approach, which is explained in the following subsection. The second approach to cope with dynamic errors is based on existing literature on subspace signatures.

%The coding schemes proposed in Sec.~\ref{sec:construction}, which are based on combining an outer MRD code with an inner optimal repair (or a locally repairable code) is not directly applicable to this active attack model. Our first solution is to correct errors during the node repair process. We will prove that for $k=2t+1$ this solution is optimal. Then we will provide a more general solution which is based on subspace signatures.

\subsection{Na\"{\i}ve Scheme for Dynamic Error Model}
In the dynamic error model, where there is no constraint on an adversarial nodes in terms of the data that it can provide during node repair as well as data reconstruction, one solution might be to design coding schemes that correct errors while node repair. One such scheme is proposed in \cite{RSRK12} for low rate codes, where a newcomer node utilizes the redundancy in the downloaded data to perform correct exact repair even in the presence of errors in the downloaded data. Next, we analyze the maximum amount of information that can be stored on the DSS employing concatenated codes proposed in Section~\ref{sec:construction} under the dynamic error model, if an error free node repair needs to be performed.

When a storage node fails, a newcomer node downloads $d\beta$ symbols from any $d$ surviving nodes $(d \geq k)$. Since there can be at most $t$ adversarial nodes present in the system, the newcomer node receives at most $t\beta$ erroneous symbols. Therefore, out of $k\alpha$ symbols of an MRD codeword, by~(\ref{eq:trade-off})  the newcomer has $(k-1)\beta + \alpha$ symbols (using the fact that the inner code is an MDS code and we perform bandwidth efficient repair). All the other $k\alpha-(k-1)\beta -\alpha=(\alpha-\beta)(k-1)$ symbols of an MRD codeword can be considered as the erased symbols. Let $\ell$ denote the number of information symbols (over $\mathbb{F}_{q^N}$) that are stored on the DSS. Then the minimum distance $\delta$ of the corresponding MRD code satisfies $\delta=k\alpha-\ell+1$. Therefore
%From MDS property of an MRD code,
we can reconstruct the entire MRD codeword and thus the data stored on the failed node, if we have
\begin{equation*}
%\label{eq:naive_repair}
\delta=k\alpha - \ell +1 \geq 2t\beta + (k-1)(\alpha - \beta) +1.
\end{equation*}
This gives us
\begin{equation}
\label{eq:bound2}
\ell \leq \alpha  + (k-2t-1)\beta.
\end{equation}
Note that the bound in (\ref{bound}) is still applicable. For $k = 2t + 1$, the right hand side expression in (\ref{eq:bound2}) is equal to that in (\ref{bound}). Therefore, this na\"{\i}ve repair scheme is optimal in terms of the capacity of DSS even in the dynamic error model. However, the difference between these bounds is monotonically increasing with $(k-2t-1)$ and the solution proposed in this section is suboptimal for general values of system parameters $k$ and~$t$.

\subsection{Subspace Signatures Approach}
\label{subsec:dynamic_errors}
%In this section, we consider the problem of designing coding schemes for DSS that work under the more general error model, dynamic error model. Under dynamic error model, an adversary can have access to at most $t$ storage nodes throughout the life span of DSS and can potentially send different data each time a node under its control is meant to send data during node repair as well as data reconstruction process. The coding scheme proposed in this paper, which is based on combining an outer MRD code with an inner bandwidth efficiently repairable (or a locally repairable code) is not directly applicable to this active attack model.

As mentioned previously, in dynamic error model an attacked node can inject a high rank error. Thus, it is desirable to restrict the rank of the aggregate error that a particular attacked node can cause in the entire system under dynamic error model. In this subsection, we propose to combine the existing literature on detecting subspace pollution with MRD codes to counter a dynamic attack. Next, we illustrate this with the help of subspace signatures proposed in \cite{zhao}.

Let us consider an $n$-nodes DSS, which employs an MRD and a bandwidth efficiently repairable code based storage scheme as explained in Section~\ref{sec:construction}. For a node $i$, content stored on it, i.e. $\mathbf{y}_i$ can be viewed as an $N \times \alpha$ matrix over $\mathbb{F}_q$. These $\alpha$ vectors of length $N$ stored on $i^{\text{th}}$ node span a subspace (column space of $\mathbf{y}_i$ when viewed as a matrix over $\mathbb{F}_q$) in $\mathbb{F}^N_{q}$ of rank at most $\alpha$. Since all elements of coding matrix and repair matrices are from $\mathbb{F}_q$, $\beta$ linear combinations sent by node $i$ for a node repair $\mathbf{y}_iV$ are nothing but, $\beta$ vectors that lie in the subspace spanned by $\mathbf{y}_i$. If we make sure that even under dynamic error model an attacked node sends vectors from the same $\alpha$-dimensional subspace of $\mathbb{F}_q^N$ during node repair and data construction, a data collector encounters at most $t\alpha$-rank error, which can be corrected with an MRD code of large enough rank distance as in the static model. Subspace signatures solve this problem of enforcing a node to send data (vectors) from the same $\alpha$-rank subspace of $\mathbb{F}_q^N$.

At the beginning, when data is encoded to be stored on a DSS, system administrator also generates subspace signatures for data stored on each node according to the procedure explained in \cite{zhao}. We assume existence of a trusted verifier, who stores all $n$ subspace signatures, one signature for each storage node. Whenever a particular node sends data during node repair or data reconstruction, the truster verifier checks the data against the stored subspace signature corresponding to that particular storage node.

For the purpose of data reconstruction, whenever a node does not pass the signature test, that node is considered an rank erasure. If $s \leq t$ nodes fail the test during data reconstruction, the data collector deals with $s\alpha$ rank erasures and $(t-s)\alpha$ rank errors. This bound on the rank of the error that the data collector encounters is explained in the following paragraph. Given that the outer MRD code has minimum rank distance $2t\alpha+1 \geq 2(t-s)\alpha + s\alpha +1$, the original data can be reconstructed without an error. %As explained in the following paragraph, rank of the error encountered by the data collector is at most $2(t-s)\alpha$ i(?? but an error can be also as a result of a nodes repair??)

Next, we argue how subspace signatures help restrict rank of the error introduced in the system by adversarial nodes under dynamic error model. Assume that node $i$ fails. Let $\mathcal{R}_i \subseteq \{1,\ldots,n\}\backslash\{i\}$ denote the set of $d$ surviving nodes that are contacted to repair node $i$. In order to repair node $i$, each node $j \in \mathcal{R}_i$ is supposed to send $\mathbf{y}_jV_{ji}$, where $V_{ji}$ is an $\alpha \times \beta$ repair matrix of node $i$ associated with node $j$. Since the data downloaded through all the surviving nodes is verified against subspace signatures, data from node $j$ passes the test if it is of the form $\mathbf{y}_j\widehat{V}_{ji}$, where $\widehat{V}_{ji}$ is an $\alpha \times \beta$ matrix, which may be different form $V_{ji}$. Note that  $\mathbf{y}_j\widehat{V}_{ji}$ is in the column space of $\mathbf{y}_j$.

If any of the surviving nodes does not pass the test, the trusted verifier begins the na\"{\i}ve repair for the failed node and the nodes that fail the test. During this na\"{\i}ve repair, entire data is downloaded from a set of $k-s$ nodes out of $d-s$ nodes that provide data for node repair in the first place and pass the subspace test. Here $s$ is the number of nodes that fail the subspace test. Note that each node of these $k-s$ selected nodes provides additional $\alpha - \beta$ symbols as it has already sent $\beta$ symbols (over $\mathbb{F}_q^{N}$). The decoding algorithm for MRD codes is run on $(k-s)\alpha$ symbols downloaded from the selected set of $k-s$ nodes. There can be at most $t-s$ adversary nodes present in the selected set of $k-s$ nodes ($s$ adversarial nodes that failed the subspace test are excluded from this process), which can contribute at most $(t-s)\alpha$ erroneous symbols. Since the rank distance of the MRD code is greater than $2(t-s)\alpha + s\alpha +1$, the decoding algorithm recovers the original file, which is used to get the data that is supposed to be stored on nodes being repaired.

In case when all the adversarial nodes pass the test, the data provided by each node $j \in \mathcal{R}_i$ is of the form

\begin{equation*}
\mathbf{y}_j\widehat{V}_{ji} = \mathbf{y}_jV_{ji} + \mathbf{y}_j(\widehat{V}_{ji}-V_{ji})
\end{equation*}
After performing exact repair process for node $i$, node $i$ stores
$\mathbf{y}_i + \mathbf{y_e}B_i$, where $B_i$ is an $t\alpha \times \alpha$ matrix over $\mathbf{F}_q$ and $\mathbf{y_e} = [\mathbf{y}_{i_1},\ldots, \mathbf{y}_{i_t}]$. Here $\{i_1,\ldots, i_t\}$ denotes the set of $t$ adversarial nodes. After the node repair, the trusted verifier generates a new subspace signature corresponding to the data stored on a node $i$ for future verification. At any point of time, the data stored on DSS can be represented as
\begin{equation}
\label{eq:anytime_dyn}
\widetilde{\mathbf{y}} = \mathbf{y} + \mathbf{y_eB},
%\widetilde{\mathbf{y}}=\left[\begin{array}{c}
%                            y_{1} \\
%                            \vdots \\
%                            \vdots \\
%                            y_{n}
%                          \end{array}
%\right] + \left[
%      \begin{array}{c}
%        D_1\\
%        \vdots\\
%        \vdots\\
%        D_n\\
%      \end{array}
%    \right]\left[\begin{array}{c}
%                            y_{i_1} \\
%                            \vdots \\
%                            y_{i_t}
%                          \end{array}
%\right]
\end{equation}
where an $i$th column of $\mathbf{B}$ is equal to $B_i$, $1\leq i\leq n$.
It is evident from (\ref{eq:anytime_dyn}) that the rank of the aggregate error in the system is at most $t\alpha$ and an MRD code with large enough minimum distance can ensure the reliable recovery of the original data.

The idea of using hash functions to provide tolerance against adversarial errors in DSS has previously been presented in \cite{ho10}. However, the error tolerance scheme of \cite{ho10} involves periodic verification of the data stored in the system which may not be practical. %Moreover the network coding based scheme of \cite{ho10} does not have an efficient decoding algorithm that a trusted verifier need to perform when a particular node fails the subspace verification.

%*******************************************************************************
%*                                                                             *
%*                 locally repairable codes                                    *
%*                                                                             *
%*******************************************************************************
\section{Error Resilience for Locally Repairable Codes }
\label{sec:locality}

In this section we consider a special case where $\alpha=\beta=~1$ and
 each node can be repaired  by accessing at most $r<k$ other nodes, where
$r$ is called \emph{locality}. Note that in this scenario there \emph{exist} $r$ nodes for repair
of a single node, and \emph{not every }set of $r$ nodes can be used for a node repair.
Such codes were considered in~\cite{GHSY11,OgDa11,PKLK12,PaDi12}.
Clearly, such codes do not have an MDS property, but they have a small repair bandwidth.
It is proved~\cite{GHSY11} that the minimum distance of such a code satisfies
\begin{equation}\label{eq:locality}
d_{\min}\leq n-k+2-\left\lceil\frac{k}{r}\right\rceil.
\end{equation}
An explicit construction of codes which attain this bound with systematic
locality is presented in~\cite{GHSY11, HCL07}.
However an explicit construction for codes that attain this bound and have all symbols locality is only known for $n=\lceil\frac{k}{r}\rceil(r+1)$~\cite{PKLK12}.
In what follows we show a new construction of codes with all symbols locality which attain bound (\ref{eq:locality})
and can also correct static errors. This construction is also based on MRD codes.

\subsection{ Construction for Optimal Locally Repairable Codes}

Let $\mathcal{C}$ be an $[N\times m, Nk, \delta=m-k+1]$
Gabidulin MRD code, where each codeword is considered as a vector of length $m$
over $\F_{q^N}$, and let $r$ be the given locality. Each symbol of an MRD
codeword will be stored in a new node. We consider two cases: $m\equiv 0 (\textmd{mod } r)$ and $m\equiv k\equiv j(\textmd{mod } r)$, for $1\leq j< r$.
In the first case we partition the set of $m$ coordinates into $\frac{m}{r}$
disjoint groups $G_i$, $1\leq i\leq \frac{m}{r} $, of size $r$, and in the second case we partition the set of $m$ coordinates into $\left\lfloor\frac{m}{r}\right\rfloor$
disjoint groups $G_i$ of size $r$ and one additional group $\widehat{G}$ of size $j$.
For each group $G_i$ we add a new parity node $p_i$
to be the sum of all the $r$ symbols in the same group.  For the group $\widehat{G}$  the new parity symbol $\widehat{p}$ is equal to the sum of the $j$ elements of $\widehat{G}$.  We denote by $C^{\textmd{loc}}$ the constructed code.

\begin{example}
\label{ex:running}
 Let $m=N=8$, $k=6$, $\delta=3$ $r=4$, and $\textbf{c}=[c_1,c_2,\ldots, c_8]$ be a codeword
of an $[8\times 8, 8\cdot 6,3]$ MRD code, $c_i\in \F_{q^8}$.
Then the corresponding codeword  of $C^{\textmd{loc}}$ is equal to $\textbf{c}^{\textmd{loc}}=[c_1,c_2,\ldots, c_8,p_1,p_2]$, where
$p_1=c_1+c_2+c_3+c_4$, and $p_2=c_5+c_6+c_7+c_8$.
\end{example}

\begin{theorem} A code $C^{\textmd{loc}}$  attains  bound~(\ref{eq:locality}), i.e., it is an $[n,k,d_{min}]$ code with $d_{min}=n-k+2-\lceil\frac{k}{r}\rceil$.
\end{theorem}

\begin{proof} As it was mentioned in Section~\ref{sec:preliminaries}, a codeword of a Gabidulin code can be considered as an evaluation of a linearized polynomial $f(x)\in \F_{q^N}[x]$ on the set of $m$ linearly independent points over $\F_q$, $\{g_1,\ldots, g_m\}$. Note that for the reconstruction the original data we need at least $k$ values  $f(a_1), \ldots, f(a_k)$, such that $a_1,\ldots, a_k \in \mathbb{F}_{q^N}$ are linearly independent over $\F_q$~\cite{OgDa11}. By the linearity of $f(x)$ we have that for any group $G\neq \widehat{G}$, the new parity $p_i$ is equal to $f(g_{i_1}+g_{i_2}+\ldots+g_{i_r})$,  for $i_s\in G$, and  for $\widehat{G}$ we have $\widehat{p}=f(g_{i_1}+g_{i_2}+\ldots+g_{i_{j}})$,  for $i_s\in \widehat{G}$. Hence, any $r$ symbols in $G_i\cup \{p_i\}$ are the evaluation of $f(x)$ in $r$ linearly independent points. Now we distinguish between two cases.
%\begin{itemize}

  \textbf{Case 1, $m\equiv 0 (\textmd{mod } r)$:}  In this case, $n=m+\frac{m}{r}$. Let $i$ and $j$ be two integers such that $k=m-r(i+1)+j$,  $0\leq i\leq \frac{m}{r}-1$, and $0\leq j\leq r-1$. Then $\frac{m}{r}-\left\lceil\frac{k}{r}\right\rceil=i$, and according to bound~(\ref{eq:locality}), $d_{\min}-1\leq m+\frac{m}{r}-r(\frac{m}{r}-(i+1))-j-\left\lceil\frac{k}{r}\right\rceil+1=r(i+1)-j+i+1=(r+1)(i+1)-j$.
       We will prove that any  $(r+1)(i+1)-j$ erasures can be corrected by this code. In other words, after this number of erasures we still have an evaluation of $f(x)$ in $k$ linearly independent points. Note that the worst case is
        when the erasures appear in the smallest possible number of groups and when the number of erasures inside a group is maximal.
        So we consider the case when all the symbols in $i$ groups are erased, and there is a group with $r+1-j$ erasures. Then the number of the remaining
       symbols which correspond to
        linearly independent points is $m-ri-(r-j)=k$.

 \textbf{Case 2, $m\equiv k\equiv j(\textmd{mod } r)$:} for $1\leq j< r$. In this case, $n=m+\left\lceil\frac{m}{r}\right\rceil$.
    Let $i$  be an integer such that $k=ir+j$, $0\leq i\leq\lfloor\frac{m}{r}\rfloor$. Since $m =\left\lfloor\frac{m}{r}\right\rfloor r+j$, then by bound~(\ref{eq:locality}), $d_{\min}-1\leq (\left\lfloor\frac{m}{r}\right\rfloor r+j+\left\lfloor\frac{m}{r}\right\rfloor+1)-(ir+j)+1-(i+1)=(\left\lfloor\frac{m}{r}\right\rfloor-i)(r+1)+1$. As in the previous case, we consider the worst case for erasures, when all the symbols in $(\left\lfloor\frac{m}{r}\right\rfloor-i)$ groups and one additional symbol are erased. Then the number of the remaining  symbols which correspond to
        linearly independent points is $m-(\left\lfloor\frac{m}{r}\right\rfloor-i)r=j+ir=k$.
%\end{itemize}
\end{proof}
\begin{example}In this example, we prove that the code from Example~\ref{ex:running} has minimum distance $4$, in other words, this code can tolerate any three erasures.
\begin{itemize}
  \item any $2$ erasures can be corrected by the MRD code, since its minimum rank distance is $3$.
  \item if two erasures are in the same group and  additional one erasure in another group,
  then we still have the values of the corresponding linearized polynomial
  $f(x)$ in the $3+4=7$ linearly independent points,
  \item  if all three erasures are in the same group, then we have values
  of $f(x)$  in $2+4=6$ linearly independent points.
\end{itemize}
 In all the cases, since $k=6$, we can correct these erasures.
\end{example}
%\vspace{0.5cm}

Note that if there is an adversarial node in such a system, then the erroneous
data spreads into the whole group containing this node because of the node repair
property, in other words, the Hamming weight of an error is equal to the size of
a group. However, the rank of such an error is one, and it can be corrected by an MRD code with $\delta\geq 3$.
We generalize this property of a code $C^{\textmd{loc}}$ in the following theorem.

\begin{theorem} If the related to $C^{\textmd{loc}}$ MRD code $\mathcal{C}$ has minimum rank distance
$\delta\geq 2t+1$ then $C^{\textmd{loc}}$ can tolerate at most $t$ static errors.
\end{theorem}

\begin{example}
Consider the code from Example~\ref{ex:running}, and let $t=1$. Its rank distance is $3$, and
the Hamming distance is $4$.  Since the group size is $5$, by using only
Hamming metric it is impossible to correct this error, however, by applying a decoding of MRD code we can correct it.
\end{example}

\begin{remark} The construction proposed in this section can be seen as a particular case of the construction of self-repairing homomorphic codes~\cite{OgDa11}, since it can also be described in terms of linearized polynomials. In addition, this construction is closely related to construction of pyramid codes~\cite{HCL07}, and the construction given in~\cite{PKLK12}, since an MRD code can be also seen as an MDS code over an extension  field.
\end{remark}

%update efficient = q-cyclic Gabidulin codes+ systematic MDS??

%*******************************************************************************
%*                                                                             *
%*                         Conclusion                                          *
%*                                                                             *
%*******************************************************************************
\section{Conclusion}
\label{sec:conclusion}

A novel concatenated coding scheme for distributed storage system was presented.
The scheme makes use of rank-metric codes, in particular, MRD codes as the first step of encoding the data.
In the second step of encoding  MDS optimal repair array codes or locally repairable codes are used.
 This construction provides resilience against static adversarial errors. Moreover, when
 using MDS optimal repair array codes this scheme is optimal in terms of the resilience
 capacity, and when using locally repairable codes this scheme is optimal in terms of the
 minimum distance. A modification of the scheme based on subspace signatures provides
 resilience against dynamic errors. The resilience of this scheme against a passive
 eavesdropper is a subject of our current research~\cite{new}.
%%%%%secrecy - current research

%\section*{Acknowledgment}

%%%%%%%%%%%%%%%%%%%%%%%%%%%%%%

\end{document}